\documentclass[12pt]{article}

\usepackage{setspace}
\usepackage[margin=1in]{geometry}

\usepackage{fullpage}
\usepackage{multirow}
\usepackage{amsmath, amsfonts, amssymb, amsthm}
\usepackage{bm}
\usepackage{tikz}
\usepackage{algorithm, algorithmic}
\usepackage{mdwlist}
\usepackage{enumerate}
\usepackage{subcaption}
\usepackage{float}
\usepackage{appendix}
\usepackage[round]{natbib}
\usepackage{authblk}
\usepackage{hyperref}
\hypersetup{colorlinks,
citecolor=blue,
linkcolor=blue,
urlcolor=black
}

\allowdisplaybreaks
\newcommand{\ignore}[1]{}

\newtheorem{theorem}{Theorem}

\theoremstyle{remark}
\newtheorem*{remark}{Remark}
\theoremstyle{definition}
\newtheorem{definition}{Definition}

\theoremstyle{definition}

\newcommand{\yvec}{\mathbf{Y}}
\newcommand{\calphahalf}{c_{\alpha/2}}
\newcommand{\cbar}{\bar{c}}
\newcommand{\ctilde}{\tilde{c}}

\newcommand{\bY}{\mathbf{Y}}
\newcommand{\by}{\mathbf{y}}
\newcommand{\calC}{\mathcal{C}}
\newcommand{\calS}{\mathcal{S}}

\newcommand{\var}{\mathrm{Var}}

\newcommand{\E}{\mathrm{E}}
\newcommand{\fcr}{\text{FCR}}

\newcommand{\wdfdr}{\text{wdFDR}}

\begin{document}


\title{Selective Sign-Determining Multiple Confidence Intervals with FCR Control}
\author{Asaf Weinstein\thanks{Department of Statistics, Sequoia Hall, 390 Serra Mall, Stanford University, Stanford, CA 94305 (E- mail: asafw@stanford.edu). To whom correspondence should be addressed. 
Partially supported by the Simons foundation under the Math+X program.} \quad Daniel Yekutieli \thanks{Department of Statistics and Operations Research, School of Mathematical Sciences, Tel Aviv University, Ramat Aviv, Tel Aviv 69978, Israel (E- mail: yekutiel@post.tau.ac.il)}}
\date{}
\maketitle

%
%
%
%




\begin{abstract}
Given $m$ unknown parameters with corresponding independent estimators, the Benjamini-Hochberg (BH) procedure can be used to classify the sign of parameters such that the expected proportion of erroneous directional decisions (directional FDR) is controlled at a preset level $q$. 
More ambitiously, our goal is to construct sign-determining confidence intervals---instead of only classifying the sign---such that the expected proportion of non-covering constructed intervals (FCR) is controlled. 
We suggest a valid procedure which adjusts
 a marginal confidence interval in order to construct a maximum number of sign-determining confidence intervals. 
%
%
We propose a new marginal confidence interval, designed specifically for our procedure, which allows to balance a trade-off between power and length of the constructed intervals, and, in fact, often enjoy (almost) the best of both worlds. 

We apply our methods to detect the sign of correlations in a highly publicized social neuroscience study and, in a second example, to detect the direction of association for SNPs with Type-2 Diabetes in GWAS data. 
In both examples we compare our procedure to existing methods and obtain encouraging results. 

%
\end{abstract}

\section{Introduction} \label{sec:intro}

Let $f$ be a known density symmetric about zero and suppose that an analyst collects independent observations $Y_i \sim f(y_i - \theta_i) \ i=1,...,m$ corresponding to unknown location parameters $\theta_i \in {\mathbb R}$. 
In many applications the analyst will highlight a subset of parameters which the data suggests as interesting, and then focus inference only on these highlighted parameters. 
A prototypical case is constructing confidence intervals for parameters $\theta_i$ corresponding to only rejected null hypotheses $H_{0i}: \theta_i=\theta_{0i},\ i=1,...,m$. 
For example, in RNA microarray one is interested in genes that are differentially expressed. 
More generally, we may consider a two-stage procedure, in which the analyst first attempts to answer a question of primary interest regarding each of the $\theta_i$; 
at the second stage, {\it only if} he was able to answer the primary question regarding $\theta_i$ with enough certainty, the analyst will pose a secondary---{ follow-up}---question regarding $\theta_i$. 
The first question is often intended to detect ``signals" of one or more types; the subsequent question may depend on the answer to the first, and is usually intended to learn about further qualities of $\theta_i$. 


In this article the primary question concerns the sign of the parameters. 
Specifically, the analyst is interested first in classifying the sign of parameters $\theta_i$ as positive ($>0$) or non-positive ($\leq 0$); 
we refer to this as {\it weak sign classification}. 
A third decision---declaring ``inconclusive data"---is allowed when an observation size is too small to infer the sign. 
The sign classification problem is important in many applications. 
For example, when comparing several drugs to a control it may be of interest to determine which are better (difference positive) and which are not (difference non-positive). 
\citet{bohrer1979multiple} and \citet{bohrer1980optimal} called the sign classification problem a multiple {\it three-decision} problem \citep[in reference to][]{neyman1977synergistic} and studied optimality of decision rules under familywise error rate (FWER) control, namely, the probability of making at least one incorrect directional decision. 
In the current paper we consider rules that control the {\it weak directional} FDR, which we define as
\begin{equation}
\wdfdr := \E\left[ \frac{V_{D}}{R_{D}\vee 1} \right]
\end{equation}
where $R_D$ is the number of parameters whose sign was classified, and $V_D$ is the number of parameters whose sign was {incorrectly} classified: non-positive parameters declared positive or positive parameter declared non-positive. 

One procedure that is known to control wdFDR is the directional-BH procedure. 
Here parameters are selected via the usual BH procedure using two-sided p-values for testing $H_{0i}:\theta_i=0$, and a directional decision is made for each selected parameter according to the sign of the estimator. 
The directional-BH procedure in fact makes a {\it strict} sign classification regarding each selected parameter, i.e., $\theta_i$ is declared negative (if $Y_i<0$) or positive (if $Y_i>0$), and still controls the expected proportion of incorrect such decisions \citep{benjamini2005false}. 
The latter is known as the {\it mixed directional} FDR \citep{benjamini1993false}, and is a stronger version of directional FDR; it may differ from the wdFDR when zero parameters are possible. 
The procedures we consider are required to make only weak directional decisions and control the wdFDR. 
To fix terminology, when we refer to a sign-classification or a directional decision, unless otherwise noted, it should be understood in the weak sense, i.e., positive or non-positive. 

While a directional decision may be of primary importance, in practice it is almost always desirable to supplement such a decision with a confidence interval. 
For example, if the data suggests that a particular drug among a candidate set improves over control, we would like to be able to say how big the difference is at least; if the difference is immaterial, prescription of that drug may not be recommended after all. 

Thus, a more general objective is to construct, upon observing $Y_i, i\leq m$, confidence intervals for a {\it subset} of the $m$ parameters, such that (i) each constructed confidence interval includes either only positive or only nonpositive values, and (ii) we have control at level $q$ over the false coverage rate \citep{benjamini2005false}
\begin{equation}
\fcr = \E\left[ \frac{V_{CI}}{R_{CI}\vee 1} \right],
\end{equation}
where $R_{CI}$ is the number of confidence intervals constructed and $V_{CI}$ is the number of non-covering confidence intervals constructed. 
A confidence interval is said to be {\it sign determining} if it includes only positive or only non-positive values. 
Correspondingly, we call a procedure with property (i) above a Selective Sign-Determining Confidence Intervals procedure, and abbreviate it {\it Selective-SDCI} hereafter. 
We say that a Selective-SDCI procedure with property (ii) above is {\it valid at level $q$}. 

The directional-BH procedure corresponds to a valid Selective-SDCI procedure, which trivially constructs the interval $(0,\infty)$ for any selected parameter declared positive, and the interval $(-\infty,0)$ (note the exclusion of zero) for any parameter declared negative. 
In the current article we first propose a more general valid procedure, that can construct {\it nontrivial} sign-determining selective confidence intervals. 
In the case of a single parameter, $Y\sim f(y-\theta)$, our procedure is very simple: 
suppose that $\calC(y; \alpha)$ is any marginal $1-\alpha$ confidence interval, i.e., $\Pr_{\theta}\left(\theta\notin \calC(Y; \alpha)\right)\leq \alpha$. 
Then construct $\calC(Y; \alpha)$ if and only if it is a subset of $(0,\infty)$ or of $(-\infty,0]$. 
In that case,
\begin{equation} \label{eq:fcr-single}
\begin{aligned}
\fcr &= \Pr \big(\{\theta \notin \calC(Y; \alpha)\} \cap \{\calC(Y; \alpha) \text{ is constructed}\} \big) \\
&\leq \Pr \big(\theta \notin \calC(Y; \alpha) \big) \leq \alpha.
\end{aligned}
\end{equation}
\citet{benjamini2005false} pointed out \eqref{eq:fcr-single} to demonstrate that selection is handled gracefully as long as multiplicity is not involved. 
A main thrust of the current work is to extend the above to the case of general $m$. 
Of course, reporting all marginal $1-q$ confidence intervals which are a subset of either $(0,\infty)$ or $(-\infty,0]$, will not lead to a valid $q$-level procedure. 
Instead, for any marginal confidence interval procedure, our method employs the FCR adjustment of \citet{benjamini2005false} to produce the largest set of {\it sign-determining} FCR-adjusted marginal confidence intervals. 

Since the choice of the {\it marginal} confidence interval procedure determines our Selective-SDCI procedure entirely, it controls both the power of the procedure as a sign-determining rule and the length (and shape) of constructed intervals. 
In line with recent work of \citet{fithian2014optimal} and \citet{tian2015selective}, we will see that in our procedure there is a trade-off between these objectives: higher power generally bears a cost of lower ``accuracy" of the constructed intervals (as measured by their length and shape). 
This leads us to derive a {\it new} {marginal} confidence interval which enables our procedure to control the trade-off. 
The corresponding procedure determines the sign of parameters according to a level-$(\psi\cdot 2q)$ directional-BH procedure for $1/2\leq \psi\leq 1$, and constructs sign-determining intervals which, loosely speaking, are longer for larger $\psi$. 

Along with the original sign problem, we offer two extensions of our procedure: the first is motivated by an example from genetics, and generalizes our procedure to the case of a two-dimensional parameter, where the primary objective is to classify the sign of the first component. 
The second extension goes beyond the sign problem: we consider detecting parameters $\theta_i>\delta$ or $\theta_i<-\delta$ where $\delta$ is some pre-specified quantity; 
to address this problem we offer a procedure which constructs selective confidence intervals such that each interval contains either only values larger than $\delta$ or only values smaller than $-\delta$. 


The paper is organized as follows. 
Section \ref{sec:review} reviews the work of \citet{benjamini2005false}, which will serve us in Section \ref{sec:ssdci} to derive a class of valid Selective-SDCI procedures. 
Section \ref{sec:mqc} presents a new marginal confidence interval, designed specifically to be used in a Selective-SDCI procedure. 
In Section \ref{sec:conf-regions} we generalize our method to construct sign-determining confidence regions for parameters with more than one dimension. 
Results from a simulation study are reported in Section \ref{sec:simulation}. 
In Section \ref{sec:example} we use our method to detect the sign of correlations in a neuroscience study. 
Proofs are generally deferred to the Appendix. 

\bigskip

{\bf Notation. }
$\calC(y;\alpha)$ is a confidence interval that covers the true value with probability at least $1-\alpha$, and should be understood as a function of both $y$ and $\alpha$ unless the context suggests otherwise. 
To emphasize the dependency on $\alpha$, we sometimes write $\{\calC(\cdot; \alpha): 0\leq \alpha \leq 1\}$ instead and call it a confidence interval {\it procedure}. 
Throughout, $f$ denotes a probability density and $F$ is the corresponding distribution function. 
We denote by $c_{\alpha}$ the $1-\alpha$ quantile of a distribution, that is, the value $F^{-1}(1-\alpha)$; $z_{\alpha}$ is used for the special case of a standard normal distribution. 
Throughout, we write $Y_{(i)}$ for the observation with $i$-th {\it largest absolute value}, i.e., $|Y_{(m)}| \leq |Y_{(m-1)}| \leq ... \leq |Y_{(1)}|$. 
Finally, for any set $B$ define $-B:=\{-x:x\in B\}$. 
We tried to minimize the use of non-standard acronyms, but avoiding them altogether would result in a tedious manuscript. 
The few important ones are: CI=Confidence Interval; wdFDR=Weak Directional-FDR; SDCI=Sign-Determining Confidence Intervals; 
BY=\citet{benjamini2005false}; QC=Quasi-Conventional; MQC=Modified Quasi-Conventional. 
The reader might find this list convenient to return to if any confusion arises.

\section{Review} \label{sec:review}

\citet[][BY hereafter]{benjamini2005false} set up a framework for selective inference when multiple parameters are considered. 
Let $\yvec=(Y_1,...,Y_m)$ be a vector of estimators where $Y_j \sim f(y_j-\theta_j)$. 
Suppose that $\calS$ is a pre-specified selection rule yielding a subset $S = \calS(\yvec) \subset \{1,...,m\}$ and that a procedure, which may depend on $\calS$ and on $S$, is used to construct confidence intervals for only the selected parameters $\{\theta_j: j\in S\}$. 
Denote by $R_{CI}$ the number of confidence intervals constructed and by $V_{CI}$ the number of non-covering confidence intervals constructed. 
Then BY define the false coverage-statement rate (FCR) to be the expected value of
\[
Q_{CI} = 
\frac{V_{CI}}{ R_{CI} \vee 1 }
\]

\noindent Thus the FCR depends on $\calS$, which specifies what subset of parameters is selected in light of the data, and on the procedure which specifies how confidence intervals are constucted for any selected subset of parameters.

Suppose that at our disposal is a marginal confidence interval procedure $\{\calC(\cdot; \alpha): 0\leq \alpha \leq 1\}$ which,  for any $\alpha\in [0,1]$, specifies a ($1-\alpha$)-level marginal confidence interval for $\theta$ based on $Y\sim f(y-\theta)$. 
That is, $\Pr_{\theta}\left( \theta \in \calC(Y ;\alpha) \right) \geq 1-\alpha$ holds for any $\alpha\in [0,1]$. 
Throughout the paper we will often refer to a marginal confidence interval {\it procedure} simply as a confidence {\it interval} and write $\calC(y; \alpha)$, where it should be understood as a function of both $y$ and $\alpha$. 
Suppose that the procedure $\calC$ satisfies the following monotonicity requirement: 

\bigskip
\noindent {\bf Requirement (MON 1)} \ $\text{For any } y \text{ and any } 0\leq \alpha, \alpha' \leq 1, \text{ if }\ \ \alpha'\leq \alpha \text{ then } \calC(y;\alpha) \subseteq \calC(y;\alpha')$. 
\bigskip

\noindent Denote $CI_i(\alpha) = \calC(Y_i; \alpha)$. 
For $m>1$, if $\calS$ is an arbitrary selection rule, then constructing $CI_i(q)$ for each $i\in S $ does not, in general, guarantee $\fcr\leq q$. 
This should be obvious from considering, for example, a rule that selects the parameter corresponding to the largest of $m>1$ independent estimators (here $R_{CI}\equiv 1)$. 
On the other hand, constructing the marginal confidence interval at level $1-q/m$ trivially ensures $\fcr\leq q$. 
Indeed, denoting by $NCI_i$ the event $\{ \theta_i\notin CI_i(q/m), i\in S \}$, we have
\begin{align*}
\fcr = \E[Q_{CI}] \leq \Pr\left( \cup_{i=1}^m NCI_i \right) \leq \Pr\left( \cup_{i=1}^m \{ \theta_i\notin CI_i(q/m) \} \right) \leq q.
\end{align*}
Yet under independence of the estimators, BY show that the Bonferroni adjustment is conservative, and a smaller increase in the confidence level is sufficient to ensure $\fcr\leq q$. 
Specifically, they prove that the FCR is controlled at level $q$ under the following scheme.

\begin{definition} \label{def:BY}
Level-$q$ BY FCR-Adjusted Selective-CI Procedure
\end{definition}

\begin{enumerate}
\item Apply the selection criterion $\calS$ to obtain $\calS(\yvec)$.
\item For each selected parameter $\theta_i, \ i \in \calS(\yvec)$, let
\begin{equation}
R_{\min}(\yvec^{(i)}) = \min_y \left\{ \left| \calS(\yvec^{(i)}, Y_i=y) \right|: i\in \calS(\yvec^{(i)}, Y_i=y) \right\},
\end{equation}
where $\yvec^{(i)}$ is the vector obtained by omitting $Y_i$ from $\yvec$.
\item For each selected parameter $\theta_i, \ i\in \calS(\yvec)$, construct the following confidence interval:
\[
CI_i\left( \frac{R_{\min}(\yvec^{(i)}) \cdot q}{m} \right).
\]

\end{enumerate}
For many selection criteria, e.g., the step-up procedure of Benjamini and Hochberg, the term $\left| \calS(\yvec^{(i)}, Y_i=y) \right|$ is constant for all $y$ such that $i\in \calS(\yvec^{(i)}, Y_i=y)$ , implying that $R_{\min}(\yvec^{(i)}) = R_{CI}$. 
In that case, to adjust the confidence intervals one simply multiplies the marginal non-coverage level $q$ by the number of parameters selected and divides by $m$.

\section{Selective-SDCI procedures that determine the sign} \label{sec:selective-SDCI} \label{sec:ssdci}
In this section we propose a general scheme to produce valid Selective-SDCI procedures, which relies on a marginal confidence interval: starting with any marginal confidence interval, we show how a valid Selective-SDCI procedure can be obtained utilizing the FCR adjustment of \citet{benjamini2005false}. 
We then turn to discuss how the choice of the marginal confidence interval affects the resulting selective procedure. 

Suppose that $\{\calC(\cdot; \alpha): 0\leq \alpha \leq 1\}$ is any marginal confidence interval procedure satisfying Requirement (MON 1) of the previous section as well as 

\bigskip
\noindent {\bf Requirement (MON 2)} \ For any $0\leq \alpha \leq 1, \ \calC(-y;\alpha) = -\calC(y;\alpha)$ and the lower boundary $l(y) = \inf\left\{ \nu: \nu \in \calC(y ;\alpha) \right\}$ is increasing in $y>0$. 
\bigskip

\noindent Define a corresponding Selective-SDCI procedure as follows.

\begin{definition} \label{def:fcr-ssdci}
Level-$q$ BY-adjusted Selective-SDCI procedure
\end{definition}

\begin{enumerate}
 \item Let $Y_{(i)}$ be the estimate with the $i$-th largest absolute value, i.e., $|Y_{(m)}| \leq |Y_{(m-1)}| \leq ... \leq |Y_{(1)}|$

 \item Denoting $CI_i(\alpha) = \calC(Y_i;\alpha)$, find 
 \begin{equation*}
 R = \max \left\{ r: CI_{(r)}\left(\frac{r \cdot q}{m}\right) \text{ is contained in $(-\infty,0]$ or in $(0,\infty)$}\right\}
 \end{equation*}
 and let $\calS^*(\yvec) = \left\{ i: |Y_i| \geq Y_{ (R)} \right\}$ be the (possibly empty) set of selected parameters. 
 \item For each $i \in \calS^*(\yvec)$, construct the confidence interval
\begin{align*}
CI_i \left(\frac{R\cdot q}{m}\right).
\end{align*}
\end{enumerate}

\begin{theorem} \label{thm:ssdci}
Suppose that $Y_i\sim f(y_i-\theta_i), \ i=1,...,m$ are independent and let $\{\calC(\cdot; \alpha): 0\leq \alpha \leq 1\}$ be a marginal confidence interval procedure satisfying Requirement (MON 1) and Requirement (MON 2). 
Then the procedure in Definition~\ref{def:fcr-ssdci} enjoys $\fcr \leq q$.
\end{theorem}

\begin{proof}
We show that the procedure of Definition \ref{def:fcr-ssdci} uses the BY FCR-adjusted confidence level for the constructed CIs, in other words, the Selective-SDCI procedure is just the BY procedure in Definition \ref{def:BY} for the selection rule $\calS^*$ in Definition \ref{def:fcr-ssdci}. 
This will finish the proof, as the level-$q$ BY procedure has $\fcr \leq q$ for any selection rule. 

It remains to show that for the procedure in Definition~\ref{def:fcr-ssdci}, $R_{\min}(\yvec^{(i)}) = R$, in other words, $|\calS^*(\bY^{(i)}, Y_i=y)|$ is constant over $y$ for all $y$ such that $i\in \calS^*(\yvec^{(i)},Y_i=y)$. 
Indeed, if this is true, then the constructed intervals use the BY-adjusted level and therefore $\fcr \leq q$.
This part is proved in the appendix. 
\end{proof}

For a given marginal confidence interval $\calC(y;\alpha)$ the procedure of Definition~\ref{def:fcr-ssdci} constructs the largest number possible of BY FCR-adjusted confidence intervals that determine the sign. 
Since the set of discoveries is determined based on the adjustment of a marginal confidence interval, our procedure is completely characterized by the choice of $\calC(y;\alpha)$. 
Therefore, this choice affects both the power of the procedure as a sign classification rule---the expected (say) number of intervals constructed---and the shape of the constructed intervals. 
In particular, using a marginal interval $\calC(y;\alpha)$ which determines the sign for relatively small values of $|y|$ will enhance power. 
On the other hand, if a marginal interval $\calC(y;\alpha)$ with relatively (to other marginal CIs) small maximum length is used, then the constructed confidence intervals will enjoy a relatively (to other marginal CIs adjusted at the same level) small maximum length, because only the confidence level is adjusted when constructing the intervals. 
The following examples describe the BY-adjusted Selective-SDCI procedure corresponding to three different choices of a marginal confidence interval. 
We will assume here that $Y_i-\theta_i$ are i.i.d. $N(0,1)$ and denote $z_{p} = \Phi^{-1}(1-p)$. 

\begin{enumerate}[(a)]
\item { Symmetric confidence interval}. 
Set $\calC(y; \alpha) = (y - z_{\alpha/2}, y + z_{\alpha/2})$. 

\noindent Since for any $\alpha\in (0,1)$ this confidence interval includes values of one sign only (and possibly zero) whenever $z_{\alpha/2} \leq |y|$, the algorithm in Definition~\ref{def:fcr-ssdci} selects the parameters corresponding to the $R = \max \left\{ r: \right.$ $\left\{ z_{ r \cdot q/(2m) } \leq |Y_{(r)}| \right\}$
largest observations. 
Now let $P_i = 2( 1-\Phi(|Y_i|) )$ be the two-sided p-value for testing $H_{0i}: \theta_i = 0$, and let $P_{(1)} \leq P_{(2)} \leq ... \leq P_{(m)}$ be the ordered p-values (note that the subscript of the order statistic has the conventional meaning for the p-values but not for the estimators). 
Then $R = \max\left\{ r: P_{(r)} \leq r \cdot q / m \right\}$ and so the selected parameters are exactly those corresponding to hypotheses rejected by the BH procedure applied at level $q$. 
The constructed confidence interval for each selected parameter $\theta_i$ is \ $CI_i = Y_i \pm z_{ R \cdot q/(2m) }$. 

\item {One-sided confidence interval} \footnote{For lack of a better term we refer to the CI in \eqref{eq:one-sided} as ``one-sided", although this name is usually reserved for a CI of the form $(y-z_{\alpha},\infty)$ or $(-\infty, y+z_{\alpha})$}. 
Take 
\begin{align} \label{eq:one-sided}
\calC(y; \alpha) = 
\begin{cases}
(-\infty, \infty), & \ -z_{\alpha} < y < z_{\alpha} \\
(0, \infty), & \ z_{\alpha} \leq y \\
(-\infty, 0], & \ y \leq -z_{\alpha}.
\end{cases}
\end{align}
For any $\alpha$ this confidence interval includes values of one sign only already when $z_{\alpha} \leq |y|$. 
Our procedure therefore selects the set of parameters corresponding to the  
$R = \max \left\{ r: z_{ r \cdot q/m } \leq |Y_{(r)}| \right\}$ $ = \max\left\{ r: P_{(r)} \leq r \right.$ $\left. \cdot (2 q) / m \right\}$
largest observations, which is the set of parameters rejected by the BH procedure when applied at level $2q$. 
The constructed confidence interval for each selected parameter $\theta_i$ is $CI_i = (0, \infty)$ if $0 < Y_i$ and $CI_i = (-\infty,0]$ if $Y_i < 0$. 

\item { Pratt's confidence interval}\footnote{The original CI suggested by Pratt treats zero ``symmetricly", whereas we append zero to the negative part of the line; \eqref{eq:pratt} is therefore slightly different from the original construction, but the difference is not essential. }. 
We can use a more sophisticated one-sided interval, 
\begin{equation} \label{eq:pratt}
\calC(y;\alpha) = 
\begin{cases}
(y-z_{\alpha}, y+z_{\alpha}),\ &\text{if \ $|y|<z_{\alpha}$} \\
(0,y+z_{\alpha}),\ &\text{if \ $z_{\alpha} \leq y$} \\
(y-z_{\alpha},0],\ &\text{if \ $y\leq -z_{\alpha}$}
\end{cases}. 
\end{equation}
This construction was suggested by \citet{pratt1961length}, who sought to minimize the expected length of a confidence interval at $\theta=0$. 
Pratt's interval still determines the sign at $z_{\alpha}$ but its length is finite when it determines the sign, as opposed to the usual one-sided interval. 
The resulting FCR-adjusted selective-SDCI procedure therefore still has $R = \max \left\{ r: z_{ r \cdot q/m } \leq |Y_{(r)}| \right\}$ and selects according to a level-$2q$ BH procedure. 
However, the constructed confidence interval for a selected parameter is now
\[
CI_i = 
\begin{cases}
(0,Y_i+z_{Rq/m}),\ &\text{if \ $z_{Rq/m}<Y_i$} \\
(Y_i-z_{Rq/m},0],\ &\text{if \ $z_{Rq/m}<Y_i$}
\end{cases}
\]
instead of the infinitely long intervals that are constructed with the plain one-sided interval \eqref{eq:one-sided}. 
\end{enumerate}
It is easy to verify that all marginal confidence intervals above are valid (i.e., have $1-\alpha$ coverage) and satisfy the two monotonicity requirements (MON 1) and (MON 2). 

For a fixed $\alpha$ we would ideally want to equip the BY-adjusted Selective-SDCI procedure with a marginal interval $\calC(y;\alpha)$ which determines the sign as early as possible, and at the same time has the smallest possible (say, maximum) length. 
Unfortunately, these two requests are incompatible: early sign determination has a price of longer confidence intervals, at least for some values of $y$. 
This is demonstrated in the examples above: the two-sided marginal interval has shortest possible maximum length, but determines the sign starting only at the $1-\alpha/2$ quantile; whereas the one-sided marginal interval determines the sign already at the $1-\alpha$ quantile, but has infinite maximum length. 
The Pratt interval improves on the length of the one-sided interval ``for free", but its length is still unbounded in $y$, which is necessary if sign determination starting at the $1-\alpha$ quantile is desired. 
Consequently, if we are to use the procedure of Definition~\ref{def:fcr-ssdci}, then a trade-off between power and maximum (potential) length of the constructed intervals is unavoidable. 

Nevertheless, we are not limited to the marginal confidence intervals in (a)-(c), in which sign determination occurs at either of the two extremes, $F^{-1}(1-\alpha/2)$ or $F^{-1}(1-\alpha)$. 
Instead of insisting on earliest possible sign determination or smallest possible maximum length, we may choose a marginal confidence interval that balances between early sign determination and maximum length. 
That is, a marginal confidence interval which determines the sign starting at a value slightly bigger than the $1-\alpha$ quantile, and in turn have maximum length that is only slightly larger than twice the $1-\alpha/2$ quantile. 
Equipped with such a marginal family, the procedure of Definition~\ref{def:fcr-ssdci} will select parameters according to a BH procedure at a level close to $2q$, while controlling the length of the constructed confidence intervals.

 \citet{benjamini1998confidence} suggested a non-equivariant marginal confidence interval which is fit for the job. 
 They assume that $Y\sim f(y-\theta)$ with $f=F'$ a unimodal, symmetric density, and obtain their Quasi-Conventional (QC hereafter) by inverting a family of acceptance regions. 
Specifically, the QC interval at $y$ is defined as the convex hull of 
\begin{equation*}
\{\theta: y\in A_{QC}(\theta)\}
\end{equation*}
where
\begin{equation} \label{eq: QC AR}
A_{QC}(\theta)=
  \begin{cases}
( \theta-\cbar, \theta + \ctilde), & 0<\theta\leq \cbar \\
(0,\theta + F^{-1}(1-\alpha + F(-\theta))), & \cbar<\theta\leq \calphahalf \\
( \theta - \calphahalf, \theta + \calphahalf ), & \calphahalf < \theta
  \end{cases}
\end{equation}
and $A(\theta) = -A(-\theta)$ for $\theta<0$. 
The acceptance region at zero is symmetric in the original construction but we take 
\begin{equation*}
A_{QC}(0) = (-\infty,c_{\alpha})
\end{equation*}
which fits our (asymmetric) definition of sign determination. 
The constants $\cbar,\ctilde$ are determined by a parameter $1/2 \leq \psi < 1$ and given by 
\begin{equation*}
\cbar = F^{-1}(1-\psi\alpha)\ \ \ \ \ \ \ \ctilde = F^{-1}(1-\alpha + F(-\cbar)).
\end{equation*}
For any $p\in [0,1]$ we write $c_p = F^{-1}(1-p)$ for the $(1-p)$-th quantile of $F$. 
The QC confidence interval determines the sign for $|y|\geq \cbar\in (c_{\alpha},\calphahalf]$ and can be shown to have maximum length $\ctilde+\calphahalf<\infty$. 
The parameter $\psi$ controls the balance between early sign determination and maximum length of the QC interval. 
For $\psi=1/2$ we have $\cbar=\calphahalf$ and the usual symmetric confidence interval obtains. 
When $\psi\to 1$, $\cbar\to c_{\alpha}$ and for any fixed $y$ the Pratt interval obtains in the limit. 
As $\psi$ increases from 1/2 to 1, sign determination occurs at a gradually earlier point at the cost of an increasing maximum length. 

Since for any $\alpha$ the QC interval with $1/2 \leq \psi < 1$ determines the sign at $\cbar<\calphahalf$, the BY-adjusted Selective-SDCI procedure using the QC interval will have more power than using the symmetric interval. 
At the same time constructed intervals will be shorter as compared to using the Pratt confidence interval, and their length never exceeds $F^{-1}(1-q' + F(-F^{-1}(1-\psi q')))$ for $q'=Rq/m$ (this is just $\ctilde+\calphahalf$ for $q=q'$). 
While the QC interval already has the features that would make our procedure balance between power and length, an improvement is in fact possible. 
Indeed, we will show that the QC interval can be slightly modified so that our procedure constructs shorter intervals at no expense.

\section{A Modified Quasi-Conventional CI} \label{sec:mqc}

In this section we present a new marginal confidence interval that adopts a feature from \citet{finner1994two} to modify the QC interval of \citet{benjamini1998confidence}. 
The idea is to take advantage of the fact that in a BY-adjusted Selective-SDCI procedure only sign-determining confidence intervals are ultimately constructed; 
hence inflating the QC confidence interval whenever it anyway includes values of opposite signs, has no cost on the one hand, and on the other hand it allows to construct shorter confidence intervals when the sign is determined. 

As in  \citet{benjamini1998confidence} we make the further assumption that $f=F'$ is a unimodal density. 
We obtain the modified Quasi-Conventional (MQC hereafter) interval by modifying the acceptance regions \eqref{eq: QC AR}. 
Hence, consider
\begin{equation} \label{eq:mqc-ar}
A_{MQC}(\theta)=
  \begin{cases}
( -\bar{c}, g(\theta) ), & 0<\theta\leq \bar{c} + \calphahalf \\
( \theta - \calphahalf, \theta + \calphahalf ), & \bar{c} + \calphahalf < \theta
  \end{cases}
\end{equation}
with $A_{MQC}(\theta) = -A_{MQC}(-\theta)$ for $\theta<0$, and $A_{MQC}(0) = (-\infty,c_{\alpha})$. 
For $1/2 \leq \psi < 1$,
\begin{equation*}
\cbar = F^{-1}(1-\psi\alpha)\ \ \ \ \ \ \ \ctilde = F^{-1}(1-\alpha + F(-\cbar)),
\end{equation*}
and
\begin{equation*}
g(\theta) = \theta + F^{-1}\{ 1 - \alpha + F(-\bar{c} - \theta) \}.
\end{equation*}
As before, $\psi$ is a parameter which controls how early the confidence interval determines the sign of $\theta$, and is chosen in advance.

The MQC interval is obtained as the convex hull of $\{ \theta: y \in A_{MQC}(\theta) \}$. 
Inverting the family of acceptance regions in~\eqref{eq:mqc-ar} is more complicated than it is for the QC acceptance regions because we need to distinguish between  three cases
(i) $0 < \psi \leq \psi_1$ 
(ii) $\psi_1 < \psi \leq \psi_2$ and 
(iii) $\psi_2 < \psi$. 
Here 
\begin{align*}
&\psi_1 = \psi_1(\alpha) \text{ is the value of } \psi \text{ such that } \tilde{c} = 2\bar{c} + c_{\alpha/2}\\
&\psi_2 = \psi_2(\alpha) \text{ is the value of } \psi \text{ such that } \tilde{c} = \bar{c} + 2c_{\alpha/2}. 
\end{align*}
From a practical point of view, however, at least when $f$ is the standard normal density, $\psi_1$ tends to be very close to $1$. 
For example, when $f$ is the standard normal density and $\alpha = 0.1$, $\psi_1 > 0.999$, and it is even closer to 1 for smaller $\alpha$. 
This means that for typical, small values of $\alpha$, unless $\psi$ is chosen extremely close to $1$, we are in case (i) above. 
For clarity we present here the confidence bounds for the first case only; a full specification of the CI, which includes the other two cases, is provided in the Supplementary Material. 
Hence, for $0 < \psi \leq \psi_1$, the convex hull of $\{ \theta: y \in A_{MQC}(\theta) \}$ is given by
\begin{align} \label{eq:mqc-ci-1}
\calC_{MQC}(y;\alpha) &=
  \begin{cases}
(-\bar{c} - c_{\alpha/2}, \bar{c} + c_{\alpha/2}), & 0\leq y<\bar{c} \\
(0, y+c_{\alpha/2}), & \bar{c} \leq y <\tilde{c} \\ 
(g^{-1}(y), y+c_{\alpha/2}), & \tilde{c}\leq y\leq g(\bar{c}+c_{\alpha/2}) \\
( \bar{c} + c_{\alpha/2}, y + c_{\alpha/2} ), & g(\bar{c}+c_{\alpha/2})<y<\bar{c}+2c_{\alpha/2} \\
(y-c_{\alpha/2}, y+c_{\alpha/2}), & \bar{c}+2c_{\alpha/2}\leq y
  \end{cases}
\end{align}
with $\calC(-y;\alpha) = -\calC(y;\alpha)$. 
In~\eqref{eq:mqc-ci-1} $g^{-1}(t)$ is well defined since $g$ is strictly increasing to $\infty$ on $-\cbar + \calphahalf < t$, and in particular on $\ctilde < t$. 
The assumption that $f$ is unimodal (and symmetric) ensures that \eqref{eq:mqc-ci-1} is indeed the convex hull of $\{\theta: y\in A_{QC}(\theta)\}$. 

\begin{remark}
The MQC interval is scale invariant in the following sense: 
if $Y\sim (\theta,\sigma^2)$ and $Y'=Y/\sigma$, and $\calC(y;\alpha)$ and $\calC'(y';\alpha)$ are the MQC confidence intervals (for any fixed $\psi$) based on $Y$ and $Y'$, respectively, then $\calC(y;\alpha) = \sigma \cdot \calC'(y/\sigma;\alpha)$. 
\end{remark}

The QC and the MQC intervals determine the sign of $\theta$ starting at exactly the same value of $|y|$, but the latter constructs shorter intervals on a subset of $\{y: |y|>\cbar\}$ at the expense of wider intervals for all $|y|<\cbar$. 
On this subset, for each of the three cases above, the lower endpoint is farther away from zero; in the last two cases---that is, when $\psi_1 < \psi$---there is a discontinuity point for the lower bound just when the confidence interval separates from zero ($y=\tilde{c}$).

\begin{figure}[h]
\begin{centering}
 \includegraphics[scale=.55]{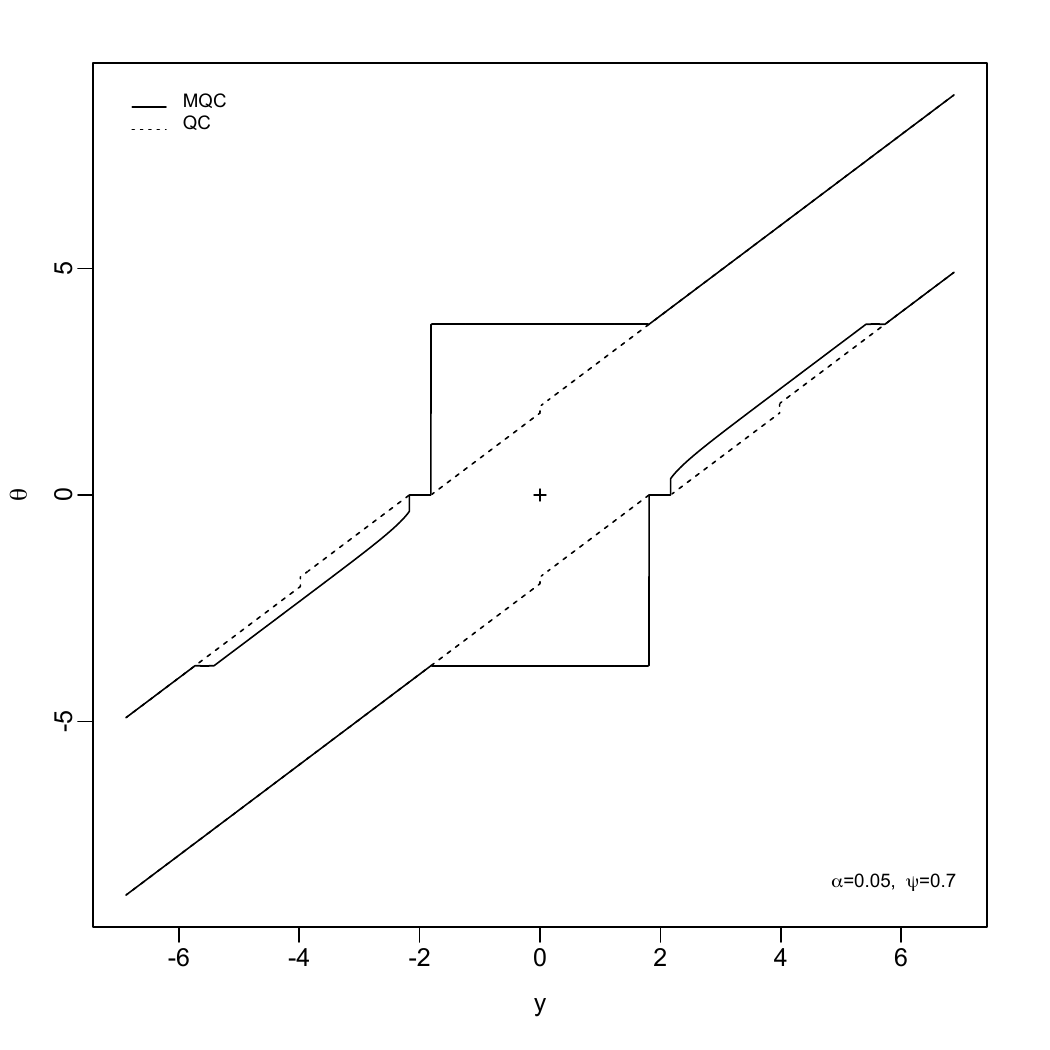}
\par\end{centering}

\caption{
MQC interval vs. the QC interval of \citeauthor{benjamini1998confidence}. 
The plot is for $\alpha = .05$ and $\psi=0.7$. 
Both confidence intervals (weakly) determine the sign of $\theta$ whenever $|y|\geq \cbar=1.81$. 
When the sign is determined, MQC bounds are farther away from zero for a range of $y$ values which begins when the confidence interval separates from zero. 
The interval around zero where the MQC confidence limits are constant in $y$ is the region where the QC interval includes both negative and positive values. 
}
\label{fig:mqc}
\end{figure}

The BY-adjusted Selective-SDCI procedure, equipped with any marginal confidence interval that satisfies requirements (MON 1) and (MON 2), has $\fcr\leq q$. 
The actual FCR level depends on the marginal confidence interval that is used. 
For the two-sided confidence interval---that is, for the BH-selected BY-adjusted procedure---\citet{benjamini2005false} show that the FCR is also lower bounded by $q/2$. 
We show a similar result for the MQC interval when the estimators are normally distributed.

\begin{theorem} \label{thm:fcr}
For independent, normally distributed estimators with a known variance, the Selective-SDCI procedure of Definition~\ref{def:fcr-ssdci} using the MQC interval with $0 < \psi < 0.9$, enjoys $\fcr \geq q/2$ if $0 < q < 0.25$.
\end{theorem}

While Theorem~\ref{thm:fcr} asserts that, under the stated conditions, using the MQC interval ensures $FCR\geq q/2$, it is typically close to $q$. 
Indeed, for standard normal observations and for $0 < \alpha < 0.25$ and $0 < \psi < 0.9$, the probability in \eqref{eq:pNCI1} of the Appendix is approximately $q$ for all $\theta$ except for a small region where it may decrease to as low as $\alpha/2$. 
For example, when $\alpha=.01$ and $\psi=.85$, as long as $|\theta|\notin (0,.48)$ and $|\theta|\notin (6.43,7.4)$, the probability in \eqref{eq:pNCI1} is at least $0.99 \alpha$. 
Hence, the inequality in \eqref{eq: NCI} can often be made much tighter and $\fcr \approx q$.
We emphasize that if the original QC interval is used in the Selective-SDCI procedure instead of the MQC interval, the FCR may fall significantly below $q/2$, as demonstrated in the simulation of Section \ref{sec:simulation}. 

\section{Sign determination by confidence regions} \label{sec:conf-regions}
 
 In this section we extend our methodology from the one-dimensional case to the case that $\boldsymbol{\theta}_i \in {\mathbb R}^k$ with $k > 1$.
 In principle, it is possible to classify $\boldsymbol{\theta}_i$ as having one of $2^k$ possible signs. 
 Here we consider a straightforward extension of the 1-dimensional case in which 
 $\boldsymbol{\theta}_i = (\boldsymbol{\theta}_{i 1}, \theta_{i 2})$, with $\boldsymbol{\theta}_{i 1} \in {\mathbb R}^{k-1}$ and $\theta_{i 2} \in {\mathbb R}^1$
 where we try to classify  $\boldsymbol{\theta}_i$ according to the sign of $\theta_{i 2}$.
In Section S.1 of the Supplementary Material we apply this methodology for classifying the sign of the genetic association of almost half a million SNPs 
and then assessing whether the effect of the SNP is recessive or dominant.

 \smallskip
Let  $\bY_i = ( Y_{i 1}, Y_{i 2})$, with independent $Y_{i 1} \sim f(y_{i 1}  - \theta_{i 1})$ and $Y_{i 2} \sim f(y_{i 2}  - \theta_{i 2})$. 
 $CI_{i 1} ( \alpha) = CI_{i 1}(\alpha; Y_{i 1})$ is a marginal $1 - \alpha$ confidence interval for $\theta_{i 1}$ and 
 $CI_{i 2} ( \alpha) = CI_{i 2}(\alpha; Y_{i 2})$ is a marginal $1 - \alpha$ confidence interval for $\theta_{i 2}$ (assume that the coverage probability is exactly $1-\alpha$, not more). 
We use $CI_{i 1} ( \alpha_1 )$ and $CI_{i 2}(\alpha_2)$ to form a $1 - \alpha_1 \cdot \alpha_2$
confidence set for $\boldsymbol{\theta}_i$,
\[
\widetilde{CI}_{i} ( \alpha_1, \alpha_2) = \{ \boldsymbol{\theta}_i : \;  \theta_{i 1} \in CI_{i 1} ( \alpha_1),  \theta_{i 2} \in CI_{i 2} ( \alpha_2) \ \}.
\]
For independent $\bY_1, ..., \bY_m$ and a selection rule that has $R_{\min}(\bY^{(i)}) \equiv R_{CI}$, 
BY show that the FCR is equal to 
$\sum_{r  = 1}^m \sum_{i = 1}^m \Pr( | {\cal S(\bY)} | = r, \tilde{NCI}_i) / r$,
where $\tilde{NCI}_i$ is the event that $\boldsymbol{\theta}_i$ is selected and 
$\boldsymbol{\theta}_i \notin \widetilde{CI}_i $.
We denote by $\tilde{NCI}_{i 1}$ the event that $\boldsymbol{\theta}_i$ is selected and  $\theta_{i 1} \notin {CI}_{i 1}$ and 
by $\tilde{NCI}_{i 2}$ the event that $\boldsymbol{\theta}_i$ is selected and $\theta_{i 2} \notin {CI}_{i 2}$.
Thus $\tilde{NCI}_i = \tilde{NCI}_{i 1} \cup \tilde{NCI}_{i 2}$,
and to evaluate FCR we express $\tilde{NCI}_i$ as the disjoint union
\[
\tilde{NCI}_i = \tilde{NCI}_{i 1} \cup ( \{\theta_{i 1} \in CI_{i 1}\} \cap  \tilde{NCI}_{i 2}).
\]
We consider selection rules ${\cal S}_2 ( \mathbf{Y}_{ \bullet 2})$ that are
 determined by only $\mathbf{Y}_{ \bullet 2} = ( Y_{1 2} ,..., Y_{m 2})$. 

\begin{definition} \label{def--BY2}
Level-$(q_1, q_2)$  FCR adjustment for selection rules determined  by $\mathbf{Y}_{ \bullet 2}$
\begin{enumerate}
\item Apply the selection criterion ${\cal S}_2$ to obtain ${\cal S}_2 ( \mathbf{Y}_{ \bullet 2})$.
\item For each selected parameter $\boldsymbol{\theta}_i, \ i \in {\cal S}_2 ( \mathbf{Y}_{ \bullet 2})$, let
\begin{equation}
R_{\min}(\mathbf{Y}_{\bullet 2}^{(i)}) = \min_t \left\{ \left| \calS(\mathbf{Y}_{\bullet 2}^{(i)}, Y_{i 2} =t) \right| : 
i \in \calS(\mathbf{Y}_{\bullet 2}^{(i)}, Y_{i 2}=t) \right\},
\end{equation}
where $Y_{\bullet 2}^{(i)}$ is the vector obtained by omitting $Y_{i 2}$ from $\mathbf{Y}_{\bullet 2}$.
\item For each selected parameter $\boldsymbol{\theta}_i, \ i \in {\cal S}_2 ( \mathbf{Y}_{ \bullet 2})$, construct the following CI:
\[
\widetilde{CI}_i\left(   q_1, \frac{R_{\min}(\mathbf{Y}_{\bullet 2}^{(i)}) \cdot q_2}{m} \right).
\]
\end{enumerate}
\end{definition}

\begin{theorem} \label{thm:fcr-2dim}
Let $\mathbf{Y}_1, ..., \mathbf{Y}_m$ be independent where $\mathbf{Y}_i=(Y_{i1},Y_{i1})$ for independent $Y_{i1}$ and $Y_{i2}$. 
Then the FCR of the level-$(q_1, q_2)$  adjusted confidence sets of Definition \ref{def--BY2} 
for  ${\cal S}_2 ( \mathbf{Y}_{ \bullet 2})$  is
\[
\fcr( \widetilde{CI}_{\bullet} ; {\cal S}; q_1, q_2) = 
q_1  +  (1 - q_1) \cdot  \fcr( CI_{\bullet 2} ;  \ {\cal S}_2 ; \ q_2)
\]
where $\fcr( CI_{\bullet 2} ;  \ {\cal S}_2 ; \ q_2)$ is the FCR of  
the level $q_2$  BY FCR-adjusted CI for ${\cal S}_2$.
\end{theorem}

\section{Simulation study} \label{sec:simulation}

We carried out two different simulations that demonstrate the performance of the BY-adjusted Selective-SDCI procedure using the MQC interval. 
Additionally, in Section S.2 
(see supplement) we report the results of a simulation in which we examined Selective-SDCI procedures under dependency. 
The first simulation illustrates the asymmetric shape of the MQC intervals and its increased power to classify the sign of parameters over the BH directional procedure. 
We took $m=200$ parameters where $\theta_1,...,\theta_{160}$ were sampled from an exponential distribution with mean $0.5$, and $\theta_{161},..., \theta_{200}$ were sampled from a $N(3,1)$ distribution. 
Each $\theta_i$ was then randomly assigned a positive or a negative sign. 
The independent observations are $Y_1,...,Y_{200}$ with $Y_i \sim N(\theta_i, 1)$.
Figure~\ref{fig:simulation} shows the constructed intervals for positive $\theta_i$ when the procedure of Definition \ref{def:fcr-ssdci} is equipped with the MQC interval ($\psi=0.85$) and applied at level $q=0.2$. 
A total of 74 sign-determining CIs were constructed, 32 of them for positive observations. 
The number of parameters selected is almost as large as the number selected with a BH directional procedure at level $2q = 0.4$ (77) and much larger than a BH directional procedure at level $q=0.2$ (55). 
Meanwhile, the MQC constructed CIs (vertical segments in the figure) are relatively short---at the most part even shorter than the symmetric FCR-adjusted confidence intervals for level-$q$ BH-selected parameters (partly thanks to the fact that more parameters are selected). 
Out of the 74 constructed CIs 14 did not cover the respective parameter (6 of which for positive observations), a proportion of $0.19$. 
The procedure using the QC interval ($\psi=0.85$) instead of MQC, constructed the same number of intervals with a false coverage proportion of $0.15$.

\begin{figure}[h]
\begin{centering}
 \includegraphics[scale=.6]{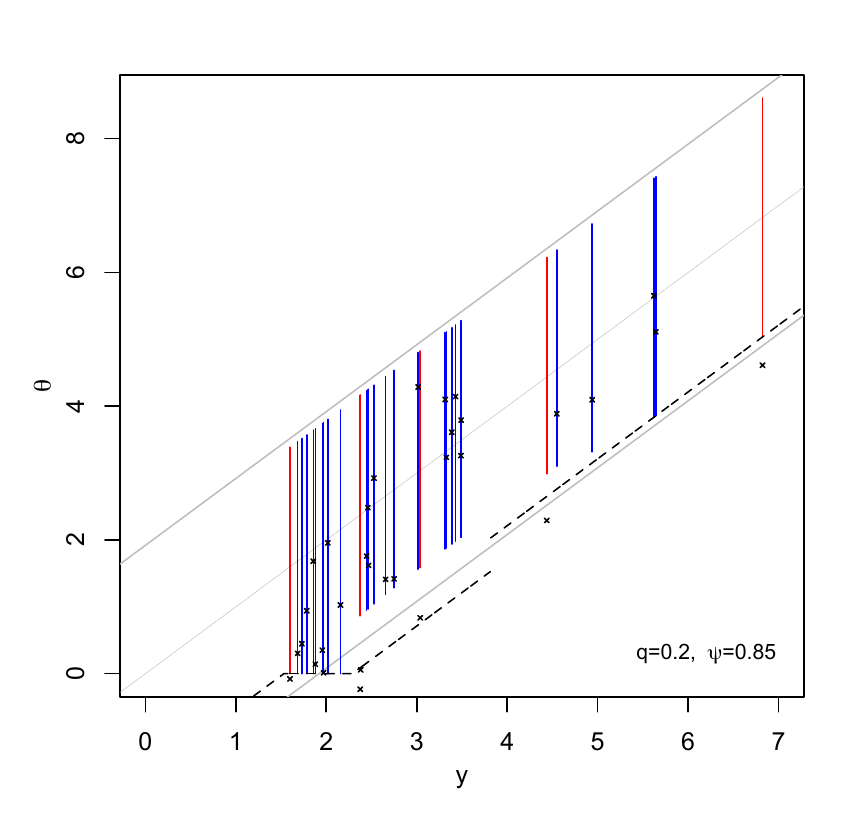}
\par\end{centering}

\caption{
MQC sign-determining confidence intervals. 
Level $q=0.2$ BY-adjusted sign-determining MQC confidence intervals were constructed for a total of 74 out of $m=300$ parameters, 14 of the confidence intervals do not cover the respective parameter. 
Vertical lines display MQC adjusted confidence intervals for the 32 positive observations: 26 of them cover the respective parameter (blue bars) and 6 of them do not (red bars). 
Diagonal line running through the origin is the identity line. 
Markers denote pairs $(Y_i,\theta_i)$. 
Solid gray lines mark the position of level 0.2 FCR-adjusted two-sided CIs for level 0.2 BH-selected parameters. 
Broken line represents the lower boundary of the constructed QC intervals.
 \label{fig:simulation}
}
\end{figure}

The second simulation compares the (actual) FCR of the MQC-equipped procedure with that of the QC-equipped procedure. 
We first sampled $\theta_1,...,$ $\theta_{300}$ from a $N(0,4)$ distribution. 
For $N=10^4$ data sets $\bY = (Y_1,...,Y_{300})$ with $Y_i \sim N(\theta_i, 1)$, we computed the false coverage proportion (FCP, denoted by $Q_{CI}$ in section \ref{sec:review}) for a level $q=0.05$ BY-adjusted Selective-SDCI procedure using the MQC interval ($\psi=0.85$) and for the same procedure using a QC interval ($\psi=0.85$). 
We used $\psi=0.85$ for both the QC and the MQC intervals so that sign determination occurs at the same value for both intervals. 
The average FCP for QC was $0.018$ ($\widehat{SD} = 1.3\cdot  10^{-4}$) and for MQC it was 0.048 ($\widehat{SD} = 2.2\cdot  10^{-4}$). 
These results confirm that, as discussed in the pervious section, the FCR when using the MQC interval is often very close to $q$ whereas it may fall below $q/2$ when using the QC interval.

\section{Detecting the sign of correlations in a social neuroscience study} \label{sec:example}

\citet{tom2007neural} carried out an experiment in an attempt to associate neural activity in the brain with behavioral ``loss aversion". 
Their study received high publicity, and the collected data was reanalyzed in \citet{poldrack2009independence} and in \citet{rosenblatt2014neuroimage}. 
The original data was made available through the \textit{OpenfMRI} initiative at \url{https://openfmri.org/dataset/ds000005} and described in detail in the paper by \citeauthor{tom2007neural}
For each of 16 subjects a behavioral loss aversion index was measured along with a neural index at each brain voxel. 
The voxel-specific correlations between behavioral index and neural index were then used to detect brain regions that are associated with loss aversion. 
\citet{rosenblatt2014neuroimage} revisited this dataset and explored different methods to construct confidence intervals which account for selection bias in reported voxels. 
Their approach is, in general, to employ a two-stage procedure where the first stage is in principle designed to detect nonzero correlations; at the second stage they construct a confidence interval for each parameter selected (rejected) at the first stage, while attempting to control the FCR below some pre-specified level. 

Specifically, one of the schemes they used is selection via the BH procedure. 
As we are interested in sign classification rather than two-sided testing, we view the BH procedure here as a directional procedure, namely, as a procedure which classifies the sign of each reported parameter as strictly positive or strictly negative. 
If willing to settle for weak (rather than strict) sign determination, our method suggests an alternative which tends to discover more parameters. 
Thus, we apply our method to the $z$-scores computed for each voxel for the Fisher-transformed correlations, and which were processed by \citet{rosenblatt2014neuroimage} and kindly made available to us. 
The concern about validity of our procedure under dependency, which is likely to be present in the current example, is mitigated by the simulations results from Section S.2 
of the Supplementary Material. 

A level 0.1 directional-BH  procedure applied to the two-sided p-values found 18,844 voxels for which a strict sign decision can be made. 
Meanwhile, a level 0.1 BY-adjusted Selective-SDCI procedure using the MQC interval with $\psi=0.85$ was able to weakly classify the sign of a total of 36,131 correlations, where for 27,117 of these a strict sign classification was made. 
For comparison, the BY-adjusted Selective-SDCI procedure using a one-sided (or Pratt's) confidence interval, which selects according to BH procedure at level 0.2, reports 43,804 parameters, all signs weakly classified. 
Hence the BH at half the level makes 57\% less discoveries, all with strict sign classification; 
whereas the MQC-equipped BY-adjusted Selective-SDCI at half the level makes only 18\% less discoveries, the majority of them with strict sign classification. 
Figure~\ref{fig:example-sign} displays the MQC confidence intervals constructed for the 33,856 correlations classified as positive, along with the QC intervals. 
The symmetric intervals corresponding to selection according to a level 0.1 BH procedure is also shown for reference. 
It is seen in the figure that for a majority of the discoveries, the lower endpoint of the MQC interval is farther away from zero than that of the QC interval, even though the latter yields the same set of discoveries. 
Note that the gap between the lower endpoint of MQC (black points in figure) and the lower endpoint of QC (gray line in figure) is largest immediately as the two intervals separate from the horizontal axis, which is exactly where we would like the gap to be largest: 
it is more important to be able to quote an endpoint farther from zero for a small detected correlation than it is for a very large detected correlation. 

We emphasize that the intervals constructed by the Selective-SDCI procedure using any of the configurations (i.e., any of the marginal confidence intervals) above are sign-determining. 
Hence, this is a partial response to the request of \citet{rosenblatt2014neuroimage}, who comment that ``it might be of interest to develop CIs that are dual to the selection methods used in neuroimaging".

Instead of detecting positive or negative correlations, it is reasonable that a researcher would be interested in detecting the large correlations, positive or negative. 
In Section S.5 
of the appendix we present an extension of the BY-adjusted Selective-SDCI procedure which allows to detect correlations $\rho_i>\rho_0$ or $\rho_i<-\rho_0$ for some pre-specified constant $\rho_0\in (0,1)$, and supplement decisions with compatible confidence intervals. 
Figure \ref{fig:example-delta} displays the 9 constructed BY-adjusted Selective confidence intervals for $\rho_0=0.2$.

\section{Discussion}\label{sec:discussion}

Selective inference refers to the general situation where the target of inference is chosen adaptively---only after seeing the data. 
We concentrated on a setup where selective inference arises in connection to multiplicity: 
the analyst collects noisy observations on a (typically large) number $m$ of unknown parameters, which he will use to first try and answer a primary question about each parameter, and second to construct CIs for only the parameters for which there was enough evidence to answer the primary question. 
Specifically, we considered the problem of detecting the sign of parameters, and supplementing each directional decision made with a CI. 
Because the same data is used for detection and for construction of the follow-up CIs, selection needs to be accounted for.

\begin{figure}[h]
\centering
  \begin{subfigure}[b]{0.35\textwidth}
    \includegraphics[width=\textwidth,height=.4\textheight]{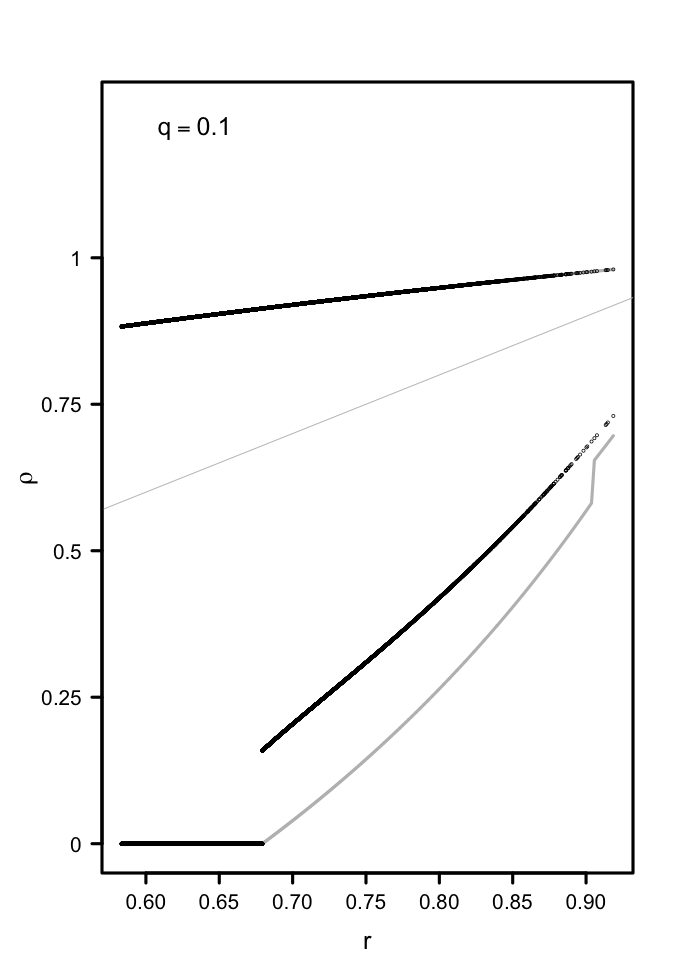}
    \caption{}
    \label{fig:example-sign}
  \end{subfigure}
  \quad
  \begin{subfigure}[b]{0.35\textwidth}
    \includegraphics[width=\textwidth,height=.4\textheight]{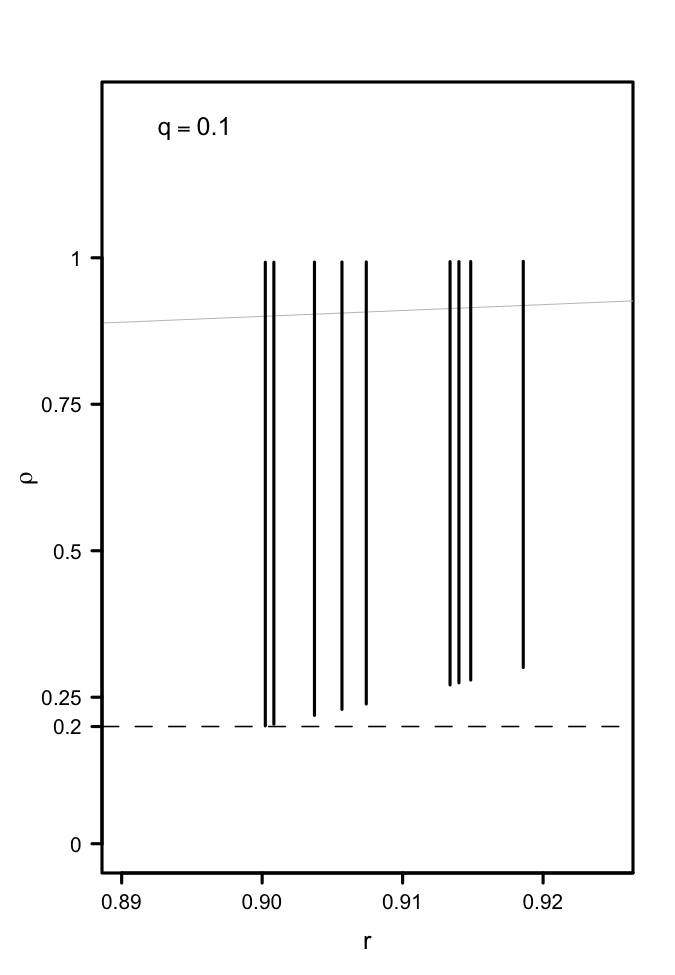}
    \caption{}
    \label{fig:example-delta}
  \end{subfigure}
  \caption{Selective CIs which determine the sign (left panel) and selective CIs which detect large correlations (right panel), data from \citet{tom2007neural}.
  CIs shown as vertical bars on right panel; on left panel only lower and upper endpoints of CI are shown. 
  In both panels observed correlations are on horizontal axis, vertical axis represents true correlation values; light gray solid line is the identity line. 
  (a) Black points correspond to MQC confidence intervals, and gray lines to QC confidence intervals, for the 33,856 correlations classified as positive. 
  The upper endpoints for the two methods coincide, while the lower endpoint of MQC is father away from zero. 
  (b) Requiring selective CIs to include only correlation values $\rho>0.2$ or only values $\rho<-0.2$, the $\text{MQC}_\delta$-equipped procedure constructs such intervals for 9 out of the original 382,362 voxels. 
No correlations $<-0.2$ were detected. 
}
\end{figure}


Requiring weak directional-FDR control at the first stage and FCR control at the second stage, a natural approach to the problem is to treat the two stages separately. 
Thus, one could first apply the directional-BH procedure to select a subset of parameters whose signs will be classified, and then construct, for each selected parameter, a CI that is  valid {\it conditionally} on selection \citep[for example by using the methods of][]{weinstein2013selection}. 
Constructing conditional CIs is appealing from various aspects, one of them being the fact that a conditional CI (usually) has the property of converging to the unadjusted CI for a large value of the observation. 
However, there is a drawback to constructing conditional CIs, namely, one cannot guarantee that the CI is compatible with the directional decision of the first stage. 
In other words, there is no conditional CI that, for all values of the parameter, with probability one includes either only positive or only non-positive values: see Section S.3 of the supplementary material, where we discuss connections to existing work on post-selection inference. 
Hence, to ensure compatibility of the follow-up CI with the directional decision, we combine the two inferential goals by requiring simply the construction of (a selective set of) {\it sign-determining} CIs. 

While the focus was predominantly on the sign problem, the approach we suggest is quite general. 
For example, in the supplement 
(Section S.5) we show how to modify the procedure of Definition \ref{def:fcr-ssdci} so that instead of sign classification, the primary goal is to detect parameters larger than $\delta$ or smaller than $-\delta$. 
In general, suppose that the data is $\bY_i\stackrel{ind}{\sim} f(\by;\boldsymbol{\theta}_i),\ i=1,...,m$ where $\bY_i \in \mathcal{Y}$ and $\boldsymbol{\theta}_i \in \Theta$. 
Let $\Theta_j \subseteq \Theta,\ j=1,...,k$ be disjoint subsets in the parameter space. 
The primary task is to detect membership of the $\boldsymbol{\theta}_i$ to any of the $\Theta_j$; the secondary task is to construct a confidence set $C_i$ for each classified parameter, such that $C_i\subseteq \Theta_j$ if $\boldsymbol{\theta}_i$ was classified to $\Theta_j$. 
For hypothesis testing $k=1$ and $\Theta_1$ is the set of alternatives; 
for the (weak) sign problem $k=2$ and $\Theta_1=(0, \infty),\ \Theta_2=(-\infty,0]$; 
the example of Section S.5 
(see supplement) corresponds to $k=2$ and $\Theta_1=(\delta,\infty),\ \Theta_2=(-\infty,-\delta)$; 
in Section S.1 
(see supplement) $\mathcal{Y}=\Theta=\mathbb{R}^2$ and $k=2$ and $\Theta_1=(-\infty,\infty)\times (0, \infty),\ \Theta_2=(-\infty,\infty)\times (-\infty,0]$. 
In principle, the extension of the procedure in Definition \ref{def:fcr-ssdci} to the general case would be to construct the maximum number of FCR-adjusted confidence sets such that each confidence set is contained in one of the subsets $\Theta_j, \ j=1,...,k$.

There are certainly remaining challenges. 
When $Y_i\sim N(\theta_i,\sigma^2)$ and $\sigma$ is unknown, 
\citet[][Section 3]{finner1994two} pointed out that a disadvantage of the CIs based on the $t$ statistic is that they are unbiased for the (natural) parameter $\theta_i/\sigma$, not $\theta_i$, and suggested an alternative CI which improves uniformly over the $t$ procedure. 
It might be of interest to try and modify \citeauthor{finner1994two}'s CI to produce an interval with similar properties as the MQC interval; for the corresponding procedure of Definition \ref{def:fcr-ssdci} to be valid, the monotonicity requirements would need to be checked, which might not be trivial. 
Another direction worth exploring is constructing sign-determining CIs for coefficients $\beta_j$ in a linear regression model; \citet{barber2016knockoff} address sign classification under directional-FDR control in the Gaussian linear model. 
To supplement such directional decisions with compatible confidence bounds is of clear practical importance.

Lastly, we think that an important issue is establishing a benchmark against which our procedure can be evaluated: while our procedure balances between power and length of constructed CIs, it is indexed by a single scalar parameter ($\psi$); it is natural to ask if more can be gained---for example, in the form of shorter CIs---when allowing more flexibility in constructing selective CIs that determine the sign. 
To be able to compare different procedures, a reasonable option is to set up a formal criterion which will take into account both power and the shape of constructed intervals.

\section*{Supplementary Material}
A supplement to this article includes an application of the methods of Section \ref{sec:conf-regions} to a genomic example; 
further simulation studies under dependency of the observations; 
a discussion of related existing work on selective inference, where we contrast the conditional approach with ours; 
a full specification of the MQC interval of Section \ref{sec:mqc}; 
and an extension of the Selective-SDCI procedure for detecting only large correlations, with an application to the example of Section \ref{sec:example}.

\section*{Acknowledgements}
Section S.1 
of the Supplementary Material makes use of data generated by the Wellcome Trust Case Control Consortium. 
A full list of the investigators who contributed to the generation of the data is available from \url{www.wtccc.org.uk}. 
Funding for the project was provided by the Wellcome Trust under award 076113.

\bibliography{/Users/assafweinstein/Dropbox/Research/References}
\bibliographystyle{plain}

\newpage
\appendix
\appendixpage

\section{Proofs}\label{app:proofs}

\subsection{A proof that $R_{\min}(\mathbf{Y}^{(i)}) = R_{CI}(\mathbf{Y})$ in Theorem \ref{thm:ssdci}}

Without loss of generality, we show that $|\calS^*(\mathbf{Y}^{(1)}, Y_1=y)|$ is constant over $y$ for all $y$ is such that $i\in \calS^*(\mathbf{Y}^{(1)},Y_1=y)$. 
Let 
\[
g(\alpha) = \inf\{ y\geq 0: \calC(y; \alpha) \text{ includes values of one sign only} \},
\]
and let $\tau(i) = g\left(\frac{i}{m}q\right),\ i=1,...,m$. 
Recall that $Y_{(i)}$ is the estimate with the $i$-th largest absolute value, hence $|{Y}_{(1)}|\geq |{Y}_{(2)}|\geq \cdots \geq |{Y}_m|$. 

Because $\calC(y; \alpha)$ satisfies the monotonicity requirements (MON 1) and (MON 2), 
$i^* = \max\left\{ i: \tau(i) \leq Y_{(i)} \right\}$ and $\tau(i)$ is a decreasing sequence. 
Define now a vector $\tilde{\mathbf{Y}} = (\tilde{\mathbf{Y}}^{(1)},\tilde{Y}_1)$, which depends on $\mathbf{Y}^{(1)}$ only, by $\tilde{\mathbf{Y}}^{(1)}=\mathbf{Y}^{(1)},\ \tilde{Y}_1 = \infty$. 
Let $\tilde{Y}_{(i)}$ the element among $\tilde{Y}_1,...,\tilde{Y}_m$ with the $i$-th largest absolute value. 
Furthermore, let
\[
\tilde{i}^* = \max\{1\leq i \leq m: \tau(i) \leq |\tilde{Y}_{(i)}|\}.
\]
We will show that if $1\in \calS^*(\mathbf{Y})$ then $i^*=\tilde{i}^*$, hence  if $1\in \calS^*(\mathbf{Y})$ then $|\calS^*(\mathbf{Y})| = i^*$, which does not depend on $y$.

First, note that $1 \in \calS^*(\mathbf{Y}) \iff \tau(\tilde{i}^*)\leq |y|$. 
Indeed, suppose that $y < \tau(\tilde{i}^*)$. 
For all $i \geq \tilde{i}^*$, $|Y_{(i)}|\leq \min(|\tilde{Y}_{(i)}|, \tau(\tilde{i}^*))$. 
Therefore, for all $i \geq \tilde{i}^*$, $|Y_{(i)}| < \tau(i)$, which together with the fact that $\tau(i)$ is decreasing implies that $i\notin \calS^*(\mathbf{Y})$ if $|Y_i| < \tau(\tilde{i}^*)$. In particular, $1 \notin \calS^*(\mathbf{Y})$.
On the other hand, if $\tau(\tilde{i}^*) \leq |y|$, then $|Y_{(\tilde{i}^*)}| = \min( |\tilde{Y}_{(\tilde{i}^*)}|, |y| ) \geq \tau(\tilde{i}^*)$, which together with the fact that $\tau(i)$ is decreasing implies that $i \in \calS^*(\mathbf{Y})$ if $\tau(\tilde{i}^*) \leq |Y_i|$. In particular, $1 \in \calS^*(\mathbf{Y})$.

To complete the proof, observe that when $\tau(\tilde{i}^*)\leq |y|$, 
(i) $|Y_{(i)}| < \tau(i)$ for $i > \tilde{i}^*$, which implies $i^* \leq \tilde{i}^*$, and 
(ii) $|Y_{(\tilde{i}^*)}| = \min( |\tilde{Y}_{(\tilde{i}^*)}|, |y| ) \geq \tau(\tilde{i}^*)$, which implies that $i^* \leq \tilde{i}^*$. 
We conclude that $i^*=\tilde{i}^*$, as required.

\subsection{Proof of Theorem \ref{thm:fcr}}

By the remark in Section \ref{sec:mqc}, it is enough to prove the theorem for the case $\var (Y_i)=1$. 
Indeed, for $\sigma^2 = \var(Y_i)$, letting $Y_i'=Y_i/\sigma$ and $\theta_i'=\theta_i/\sigma$ we have that 
$\theta_i \notin \calC(Y_i;\alpha) \iff \theta'_i \notin {\calC'}(Y_i'; \alpha)$
where $\calC(y;\alpha)$ and $\calC'(y';\alpha)$ are the MQC CIs corresponding to the distributions of $Y$ and $Y'$, respectively. 
Therefore the FCR of the procedure defined for the $Y_i$ (w.r.t. the $\theta_i$) is the same as the procedure defined for the $Y_i'$ (w.r.t. the $\theta_i'$). 

First we claim that for $\psi < 0.9$, the MQC interval is given by \eqref{eq:mqc-ci-1} for all $0 < \alpha < 0.25$. 
We need to check that $\psi_1 > 0.9$ for all $0 < \alpha < 0.25$. 
It can be verified that $\psi_1$ is a decreasing function of $\alpha$ on $0 < \alpha < 0.25$, and we have $\psi_1=0.978 > 0.9$, which together imply that $0.9 < \inf\{ \psi_1: 0 < \alpha < 0.25 \}$ as required.

Let $0 < \alpha < 0.25$ and $0 < \psi < 0.9$. 
We now consider a single parameter, $\theta$, and a corresponding estimator $Y\sim N(\theta,1)$, and show that  the probability that a sign-determining non-covering confidence interval is constructed for $\theta$, is no less than $\alpha/2$ for all $\theta$. 
Formally, let $NCI$ be the event that $CI:=\calC_{MQC}(Y;\alpha)$ 
(i) determines the sign, i.e., does not include values of opposite signs and 
(ii) does not include the true value $\theta$. 
Then we show that $\Pr_{\theta}(NCI) \geq \alpha/2$ for all $\theta$. 
Since for the MQC interval, sign determination occurs if and only if $|Y|\geq \cbar$, we have 
\begin{equation} \label{eq:pNCI1}
\Pr_{\theta}(NCI) = \Pr_{\theta}(|Y|\geq \cbar, \ \theta \notin CI).
\end{equation}
If the confidence interval were obtained simply by inverting the $1-\alpha$ acceptance regions $A(\theta)$ in \eqref{eq:mqc-ar} the event $\theta \notin CI$ could be replaced by $Y\notin A(\theta)$; 
however, the confidence interval is obtained by taking the convex hull of the inverse set, in which case it is possible that $Y \notin A(\theta)$ and yet $\theta \in CI$. 
We can overcome this difficulty by considering the ``effective" acceptance regions, $\bar{A}(\theta)$, which take into account the fact that the convex hull of $\{ \theta: Y\in A(\theta) \}$ is taken, in that $CI = \{ \theta: Y \in \bar{A}(\theta)\}$ (here without the convex hull). 
Denoting by $l(\theta)$ and $u(\theta)$ the lower and upper endpoints of $A(\theta)$, respectively, and denoting by $\bar{l}(\theta)$ and $\bar{u}(\theta)$ the lower and upper ends of $\bar{A}(\theta)$, respectively, it holds that $\bar{l}(\theta) = \max\{ u(\tilde{\theta}): \tilde{\theta}\leq \theta \}$ and $\bar{u}(\theta) = \min\{ l(\tilde{\theta}): \tilde{\theta}\geq \theta \}$. 
Explicitly,
\begin{equation} \label{eq:mqc-e-ar}
\bar{A}(\theta)=
  \begin{cases}
(-c_{\alpha/2}, c_{\alpha/2}), & \theta=0 \\
( -\bar{c}, \ctilde ), & 0< \theta \leq  \ctilde - \cbar \\
( -\bar{c}, g(\theta) ), & \ctilde - \cbar<\theta\leq \bar{c}+c_{\alpha/2} \\
( \theta - c_{\alpha/2}, \theta + c_{\alpha/2} ), & \bar{c} + c_{\alpha/2} < \theta
  \end{cases}
\end{equation}
with $\bar{A}(\theta) = -\bar{A}(-\theta)$ for $\theta<0$ and where $g(\theta) = \theta + F^{-1}\{ 2 - \alpha - F(\bar{c} + \theta) \}$.

\smallskip
\noindent Now we can write 
\begin{equation} \label{eq: pNCI2}
\Pr_{\theta}(NCI) = \Pr_{\theta}(|Y|\geq \cbar, \ Y\notin \bar{A}(\theta)),
\end{equation}
and we note that for $0 < \theta < \bar{c}+c_{\alpha/2}$, $(-c,c) \subset \bar{A}(\theta)$, hence $\Pr_{\theta}(NCI) = \Pr_{\theta}(Y\notin \bar{A}(\theta))$. 
For $\theta=0$, this is exactly $\alpha$.

\noindent For $0 < \theta < \ctilde - \cbar$, $\Pr_{\theta}(Y\notin \bar{A}(\theta)) = \Pr_{\theta}(Y\notin (-\cbar,\ctilde))$, which is minimized at $\theta = (\ctilde-\cbar)/2$. 
In order that $\Pr_{(\ctilde-\cbar)/2}(Y\notin \bar{A}(\theta))$ be less than $\alpha/2$, in which case $\Pr_{(\ctilde-\cbar)/2}(NCI) < \alpha/2$, it must hold that $\ctilde + \cbar > 2c_{\alpha/4}$. 
We claim that this cannot be the case. 
Hence, for any $\alpha$, let $\psi^*$ be the value of $\psi$ for which $\ctilde + \cbar = 2c_{\alpha/4}$. 
Then for a fixed $\alpha$, $\psi < \psi^*$ implies that $\ctilde + \cbar < 2c_{\alpha/4}$. 
Now, it can be verified that $\lim_{\alpha \to 0}\psi*>0.9$ (but $\lim_{\alpha \to 0}\psi*<0.94$) and that $\psi^*$ is an increasing function of $\alpha$ on $0 < \alpha < 0.25$, which imply that $\psi^* > 0.9$ for all $0 < \alpha < 0.25$. 
It follows that $\ctilde + \cbar < 2c_{\alpha/4}$ for all $0 < \alpha < 0.25$, and we conclude that $\Pr(NCI) \geq \alpha/2$ also for $0 < \theta < \ctilde - \cbar$. 

\noindent For $\ctilde - \cbar<\theta\leq \bar{c}+c_{\alpha/2}$, $A(\theta) = \bar{A}(\theta)$, and since $\Pr_{\theta}(Y\in A(\theta)) = 1-\alpha$, we have that $\Pr_{\theta}(NCI) = \alpha$. 

\noindent Finally, for $\theta > \bar{c} + c_{\alpha/2}$ we have $\Pr_{\theta}(NCI) = \Pr_{\theta}(|Y| > \cbar, \ |Y|>\theta + c_{\alpha/2}) \geq \alpha/2$. \\
In any case, $\Pr_{\theta}(NCI)$ does not drop below $\alpha/2$.

\smallskip
To evaluate the FCR, we follow a computation similar to that in BY. 
Let $0 < q < 0.25$ and $0 < \psi < 0.9$. 
Denote by $CI_i(\alpha) = \calC_{MQC}(Y_i; \alpha)$ a level $1-\alpha$ MQC interval using parameter $\psi$, and by $\cbar(\alpha) = \Phi^{-1}(1-\psi\cdot\alpha)$ the value of the quantity $\cbar$ associated with it.
Furthermore, let $C_k^{(i)}=\{ Y^{(i)}: R_{\min}(Y^{(i)}) = k \}$. 
For the selective-SDCI procedure of Definition \ref{def:fcr-ssdci} $R_{\min} = R_{CI}$, in which case BY show that
\begin{equation}
\fcr = \sum_{i=1}^m \sum_{k=1}^m \frac{1}{k}\Pr\left\{C_k^{(i)},\  i\in \calS(\bY),\ \theta_i\notin CI_i\left(\frac{k\cdot q}{m}\right)\right\}.
\end{equation}
Using the fact that $i \in \calS(\mathbf{Y})$ if and only if $|Y_i|\geq \cbar\left(\frac{R_{CI}\cdot q}{m}\right)$, we can replace the right hand side of the last equality by
\begin{align}
 &= \sum_{i=1}^m \sum_{k=1}^m \frac{1}{k}\Pr\left\{C_k^{(i)},\  |Y_i|\geq \cbar\left(\frac{k\cdot q}{m}\right),\ \theta_i\notin CI_i\left(\frac{k\cdot q}{m}\right)\right\} \\
 &= \sum_{i=1}^m \sum_{k=1}^m \frac{1}{k}\Pr\left\{ C_k^{(i)} \right\} \times \Pr \left\{ |Y_i|\geq \cbar\left(\frac{k\cdot q}{m}\right),\ \theta_i\notin CI_i\left(\frac{k\cdot q}{m}\right)\right \} \\
& \geq \sum_{i=1}^m \sum_{k=1}^m \frac{1}{k}\Pr\left\{ C_k^{(i)} \right\} \times \frac{kq}{2m} \label{eq: NCI} \\ 
&= \frac{q}{2}
\end{align}
where inequality~\eqref{eq: NCI} follows from the preceding part of the proof as $\frac{k\cdot q}{m} \leq q<0.25$.

\subsection{ Proof of Theorem \ref{thm:fcr-2dim}}

Beginning with an expression for FCR as appears in \citet{benjamini2005false},
\begin{eqnarray}
\lefteqn{ \fcr( \widetilde{CI}_{\bullet} ; {\cal S}; q_1, q_2) = \sum_{r  = 1}^m \sum_{i = 1}^m \frac{1}{r} \cdot \Pr( | {\cal S}_2 | = r, \tilde{NCI}_i)}   \nonumber \\
& = & \sum_{r  = 1}^m \sum_{i = 1}^m \frac{1}{r} \cdot 
\{ \ \Pr (  \ | {\cal S}_2 | = r,  \ \tilde{NCI}_{i 1})   +
\Pr (  \ | {\cal S}_2 | = r,  \  \theta_{i 1} \in CI_{i 1},  \ \tilde{NCI}_{i 2}  ) \  \} \nonumber \\
& = & \sum_{r  = 1}^m \sum_{i = 1}^m \frac{1}{r} \cdot 
\Pr (  \ | {\cal S}_2 | = r,   \  i \in {\cal S}_2 , \ \theta_{i 1} \notin CI_{i 1} )  \nonumber \\
&  &  +  \ \sum_{r  = 1}^m \sum_{i = 1}^m \frac{1}{r} \cdot 
\Pr (  \ | {\cal S}_2 | = r,  \  \theta_{i 1} \in CI_{i 1},  \ \tilde{NCI}_{i 2}  )   \nonumber \\
& = & \sum_{r  = 1}^m \sum_{i = 1}^m \frac{1}{r} \cdot 
\Pr (   \theta_{i 1} \notin CI_{i 1} ) \cdot \Pr (  \ | {\cal S}_2 | = r,   \  i \in {\cal S}_2  )  \nonumber \\
&  &  +  \ \Pr (   \theta_{i 1} \in CI_{i 1} ) \cdot
\sum_{r  = 1}^m \sum_{i = 1}^m \frac{1}{r} \cdot 
\Pr (  \ | {\cal S}_2 | = r,   \ \tilde{NCI}_{i 2}  )   \nonumber  \\
& = & q_1 \cdot \sum_{r  = 1}^m   \frac{1}{r} \cdot \sum_{i = 1}^m 
\Pr (  \ | {\cal S}_2 | = r,   \  i \in {\cal S}_2  )  +  \ (1 - q_1) \cdot  \sum_{r  = 1}^m \sum_{i = 1}^m \frac{1}{r} \cdot 
\Pr (  \ | {\cal S}_2 | = r,   \ \tilde{NCI}_{i 2}  )   \nonumber 
\end{eqnarray}
To complete the proof, note that for any  ${\cal S}_2$,
$\sum_{i = 1}^m  \Pr (  \ | {\cal S}_2 | = r,   \  i \in {\cal S}_2  )   = r \cdot  \Pr (  \ | {\cal S}_2 | = r)$,
and that 
$\fcr( CI_{\bullet 2} ;  \ {\cal S}_2 ; \ q_2) 
= \sum_{r  = 1}^m \sum_{i = 1}^m \Pr( | {\cal S}_2 | = r, \tilde{NCI}_{\bullet 2}) / r$.

\end{document}